\documentclass[a4paper,UKenglish,cleveref, autoref, thm-restate]{lipics-v2021}

\title{The 2-Attractor Problem is NP-Complete}

\author{Janosch Fuchs}{RWTH Aachen University, Aachen, Germany}{fuchs@algo.rwth-aachen.de}{https://orcid.org/0000-0003-3993-222X}{}
\author{Philip Whittington}{ETH Zürich, Zürich, Switzerland
\footnote{Parts of this work were produced as part of the author's Master's thesis at RWTH Aachen University.}}
{philip.whittington@rwth-aachen.de}{https://orcid.org/0009-0005-0910-6826}{}
\authorrunning{J. Fuchs and P. Whittington}

\Copyright{Janosch Fuchs and Philip Whittington} 

\ccsdesc[500]{Theory of computation~Problems, reductions and completeness}
\ccsdesc[100]{Theory of computation~Approximation algorithms analysis}
\ccsdesc[300]{Theory of computation~Data compression} 

\keywords{String attractors, dictionary compression, computational complexity} 

\category{} 

\relatedversion{} 

\nolinenumbers 

\EventEditors{Olaf Beyersdorff, Mamadou Moustapha Kant\'{e}, Orna Kupferman, and Daniel Lokshtanov}
\EventNoEds{4}
\EventLongTitle{41st International Symposium on Theoretical Aspects of Computer Science (STACS 2024)}
\EventShortTitle{STACS 2024}
\EventAcronym{STACS}
\EventYear{2024}
\EventDate{March 12--14, 2024}
\EventLocation{Clermont-Ferrand, France}
\EventLogo{}
\SeriesVolume{289}
\ArticleNo{41}

\newcommand{\angles}[1]{\langle #1\rangle}
\DeclareMathOperator{\col}{col}

\usepackage{tikz}

\begin{document}

\maketitle

\begin{abstract}
A \(k\)-attractor is a combinatorial object unifying dictionary-based compression.
It allows to compare the repetitiveness measures of different dictionary compressors such as Lempel-Ziv 77, the Burrows-Wheeler transform, straight line programs and macro schemes. 
For a string \( T \in \Sigma^n\), the \(k\)-attractor is defined as a set of positions \(\Gamma \subseteq [1,n]\), such that every distinct substring of length at most \(k\) is covered by at least one of the selected positions. 
Thus, if a substring occurs multiple times in \(T\), one position suffices to cover it. 
A 1-attractor is easily computed in linear time, while
Kempa and Prezza [STOC 2018] have shown that for \(k \geq 3\), it is NP-complete to compute the smallest \(k\)-attractor by a reduction from \(k\)-set cover. 

The main result of this paper answers the open question for the complexity of the 2-attractor problem, showing that the problem remains NP-complete. 
Kempa and Prezza's proof for \(k \geq 3\) also reduces the 2-attractor problem to the 2-set cover problem, which is equivalent to edge cover, but that does not fully capture the complexity of the 2-attractor problem. 
For this reason, we extend edge cover by a color function on the edges, yielding the colorful edge cover problem.
Any edge cover must then satisfy the additional constraint that each color is represented. 
This extension raises the complexity such that colorful edge cover becomes NP-complete
while also more precisely modeling the 2-attractor problem. 
We obtain a reduction showing \(k\)-attractor to be NP-complete and APX-hard for any \(k \geq 2\). 
\end{abstract}

\section{Introduction}
Compressing text without losing information is usually achieved by exploiting structural properties of the text. 
For example, dictionary-based compression works by removing redundancy resulting from repeatedly occurring substrings. 
Thus, any measurement capturing the repetitiveness of strings is directly related to the performance of dictionary compression techniques. 

In fact, Kempa and Prezza show in \cite{DBLP:conf/stoc/KempaP18} that the solutions of famous compression algorithms, like Lempel-Ziv 77, the Burrows-Wheeler transform or straight-line programs, are approximations of a certain measurement for repetitiveness, the so called \textit{string attractor}, which was introduced in \cite{DBLP:journals/corr/attractors}. 
These results are extended by Kempa and Saha \cite{DBLP:conf/soda/KempaS22} to include the LZ-End compression algorithm proposed by Kreft and Navarro \cite{lzend}, and Kempa and Kociumaka \cite{DBLP:journals/cacm/KempaK22} apply string attractors to resolve the Burrows-Wheeler transform conjecture.

Further, Kempa and Prezza \cite{DBLP:conf/stoc/KempaP18} build a universal data structure based on string attractors supporting random-access on any dictionary compression scheme. 
In \cite{universalindexing}, Navarro and Prezza improve data access when a string attractor is known, showing that the compression-based string attractors suffice to support fast indexed queries, that is, searching for all occurrences of a given pattern in a text. 
Thus, string attractors are not only the basis of dictionary compression, they also allow for a universal indexing data structure that works on top of every dictionary compression scheme. 
Christiansen et al. continue this work in \cite{fastindexing} by constructing an efficient indexing algorithm that is based on the underlying string attractor, but does not need to explicitly compute it, achieving optimal time results for locating and counting indices. 

For a string of length \(n\), an attractor is a set of positions \(\Gamma \subseteq [1,n]\) covering all distinct substrings, that is, every distinct substring has an occurrence crossing at least one of the selected positions. 
If only distinct substrings up to a certain length \(k\) need to be covered, we speak of the \(k\)-attractor. 
As an example, for the string
\[T = a\underline{b}bcab\underline{c}c\underline{a}c \text{,}\]
position 2 covers the substrings \(b\), \(ab\) and \(bb\), and the set of markings \(\Gamma = \{2,7,9\}\) forms a 2-attractor for \(T\). 
It is not a valid 3-attractor and therefore not an attractor because the substrings \(bca\) and \(cab\) are not covered.
Adding either 4 or 5 to \(\Gamma\) results in an attractor, and there is no smaller attractor for \(T\).

In \cite{DBLP:conf/stoc/KempaP18} Kempa and Prezza also show that computing the smallest \(k\)-attractor is NP-complete for any \(k \geq 3\), by giving a reduction from \(k\)-set cover, and extend this proof to non-constant \(k\), especially for \(k = n\).
On the other hand, the problem is trivially solvable in polynomial time for \(k = 1\) by a greedy algorithm. 
The complexity for \(k = 2\) was raised as an open problem. 

Variants of the problem have been introduced, such as the sharp \(k\)-attractor in \cite{DBLP:conf/esa/KempaPPR18}, which only considers distinct substrings of length exactly \(k\), and the circular attractor in \cite{DBLP:journals/tcs/MantaciRRRS21}, which also requires to cover circular substrings. 
Interestingly, the sharp \(k\)-attractor problem can also be reduced from and to \(k\)-set cover and is therefore NP-complete for \(k \geq 3\). However, the sharp 2-attractor problem reduces to 2-set cover, which is equivalent to edge cover and therefore solvable in polynomial time. Notably, the sharp variant does not exhibit the same gap as the \(k\)-attractor problem. 

Mantaci et. al \cite{DBLP:journals/tcs/MantaciRRRS21} take a combinatorial approach and study the attractor sizes of infinite families of words, such as Sturmian words, Thue-Morse words and de Bruijn words.
This lead to Schaeffer and Shallit \cite{DBLP:journals/corr/automatic} raising the definition of the string attractor problem to infinite words by considering automatic sequences and computing attractors of every finite prefix.
Further work on attractors of infinite words, especially those generated by morphisms, has been done by Restivo, Romana, and Sciortino \cite{DBLP:journals/corr/infinite}, Gheeraert, Romana, Stipulanti \cite{DBLP:journals/corr/kbonacci}, and Dvořáková \cite{DBLP:journals/corr/episturmian}.

Akagi, Funakoshi, Inenaga \cite{Akagi2021SensitivityOS} analyze the sensitivity of an attractor, i.e.,  how much can editing a single position change the result.
Bannai et. al \cite{DBLP:conf/esa/Bannai0IKKN22} formulate the attractor and other dictionary compressors as instances of the maximum satisfiability problem and present computational studies showing that an attractor can be computed in reasonable time with this approach. 

A more recent development in the field of data compression is the \emph{relative substring complexity} measure \(\delta\) \cite{relativeSubstringComplexity}, which counts the number of different substrings of length \(l\) and scales it by \(l\).
It is efficient to compute and also smaller than the size of the optimal string attractor by up to a logarithmic factor, but allows to use the results on string attractors with that overhead.

Our result is obtained by a technique in which we make a problem \emph{colorful} by adding colors to its edges (or vertices) and requiring all colors to appear in a solution. 
This idea or similar versions of it are spread out throughout the literature and thus also known by other names such as \emph{labeled}, \emph{rainbow}, \emph{tropical}, or \emph{color-spanning}.
It is mostly used in relation to paths in the graph \cite{Bentert_Kellerhals_Niedermeier_2023, temporalcolorful, akbari} and matching problems \cite{DBLP:journals/ipl/Monnot05, BUSING2018245, COHEN2017219, BEREG201926}.

\subsection{Our Contributions}
The main result of this paper answers the open question for the complexity of the $2$-attractor problem, showing that the problem remains NP-complete and thus closing the last remaining gap in the complexity analysis of the \(k\)-attractor problem. 

We introduce a more general version of the \(k\)-attractor, where the input is a set of strings and each distinct substring only needs to be covered by an attractor position in at least one of the input strings. 
We call this a \(k\)-set attractor and show that it can be simulated by a single \(k\)-attractor, showing the equivalence of the two problems. 
The ability to construct multiple strings as input makes the proof of the main result more convenient. 

Kempa and Prezza \cite{DBLP:conf/stoc/KempaP18} give reductions for \(k\)-attractor from \(k\)-set cover and to \(k(k+1)/2\)-set cover, placing 2-attractor between 2-set cover and 3-set cover. 
A \(k\)-set cover reduction is also used in \cite{DBLP:conf/esa/KempaPPR18} to show that sharp 2-attractor is solvable in polynomial time.
Essentially, both the \(k\)-attractor and the less restrictive sharp \(k\)-attractor problem are reduced from \(k\)-set cover, indicating that some of the complexity of the \(k\)-attractor is lost in the process.
For this reason, we further investigate the relation between 2-attractor and 2-set cover problem, which is equivalent to the edge cover problem.

Thus, we extend the edge cover problem with a color function on the edges, yielding the colorful edge cover problem. 
Any edge cover must then satisfy the additional constraint that each color is represented. 
The complexity of edge cover combined with the additional complexity from the colorfulness condition raises the complexity of colorful edge cover to NP-complete, although the problems each on their own are solvable in polynomial time. 
The hardness is shown by a reduction from a variable-bounded SAT variant. 
The colors are used to model a truth assignment of the variables and the edge cover condition verifies that this assignment is satisfying.
The colorful edge cover problem captures the constraints of the 2-attractor problem more tightly.
We later use this result to extend the structure of the reduction to the \(k\)-attractor problem. 
With the reduction we not only show the NP-hardness of the 2-attractor problem, we also obtain an APX-hardness result.

Our paper is organized as follows, we introduce in \Cref{sec:SetString} the set of strings attractor problem and discuss its relation to the already introduced variations, showing that it is as hard as the classical \(k\)-string attractor problem. 
In \Cref{sec:colorful} we define the colorful edge cover problem and show that it is NP-complete. 
The main result, the NP-completeness of the \(2\)-attractor problem, is presented in \Cref{sec:2attractor}. 
Afterwards, in \Cref{sec:APX}, we discuss the implicit APX-hardness that results from our reduction. 

\section{Attractors of Strings, Circular Strings, and Sets of Strings}
\label{sec:SetString}

We start with the definition of the \(k\)-attractor and explain already introduced variations before we define the \(k\)-set attractor. 
Afterwards, we discuss how to solve the corresponding problems in polynomial time, if there exists an algorithm that solves one of the problems in polynomial time.
Thus, we show that the problems are equivalent in their complexity. 

\begin{definition}[\(k\)-attractor \cite{DBLP:conf/stoc/KempaP18}] 
  A set \(\Gamma \subseteq [1,n]\) is a \(k\)-attractor of a string \(T \in \Sigma^n\) if every substring \(T[i\dots j]\) such that \(i \leq j < i+k\) has an occurrence \(T[i'\dots j']\) with \(j'' \in [i'\dots j']\) for some \(j'' \in \Gamma\).
\end{definition}

A solution \(\Gamma\) is called a \textit{string attractor} or simply \textit{attractor} if \(k=n\). 
The corresponding optimization and decision problems are the minimum-\(k\)-attractor problem and the \(k\)-attractor problem. 

For the circular \(k\)-string attractor problem, the input string is circular, resulting in additional substrings starting at the last letters of the input string and continuing at the beginning. 
Thus, there are more substrings that need to be covered compared to the \(k\)-string attractor problem. 
However, these additional substrings can make the solution smaller. 

It is convenient for the proof of \Cref{theorem:att} to allow \(k\)-attractors over multiple strings. 
This does not change the problem much, in fact \(k\)-set attractors can be easily simulated by unique delimiter symbols. 
We formally define a set attractor over \(m\) strings of possibly different lengths \(n_1, \dots n_m\) as a set of tuples, with the first entry of the tuple denoting the string and the second entry denoting the position.

\begin{definition}[\(k\)-set attractor] 
  A set \(\Gamma \subseteq \bigcup_{x=1}^m \bigcup_{y=1}^{n_x} \{(x,y)\} \) is a \(k\)-set attractor of a set of \(m\) strings \(\mathbf{T} = \{T_1, \dots T_m\}\) with \(T_x \in \Sigma^{n_x}\) if every substring \(T_x[i\dots j]\) such that \(i \leq j < i+k \) has an occurrence \(T_{x'}[i'\dots j']\) with \(j'' \in [i'\dots j']\) for some \((x', j'') \in \Gamma\).
\end{definition}

A string is a circular string cut once, applying more cuts gives a set of strings.
Therefore, if we can describe the behavior of attractors under cuts, we can show
those three to be the same. 
The equivalence of circular \(k\)-attractors and \(k\)-attractors is already shown in \cite{DBLP:journals/tcs/MantaciRRRS21}. 
We extend the idea of adding a delimiter, which must be part of any solution, to remove the impact of circularity and thereby acting as a cut, i.e., substrings that stretch over the delimiter are already covered. 
The remaining substrings are then the same as the substrings of two separate strings.

\begin{lemma}\label{lemma:setequiv}
    An algorithm solving the \(k\)-attractor problem can solve the \(k\)-set-attractor problem with linear overhead in the input size, and vice versa.
\end{lemma}
\begin{proof}
    Given a set of strings \(\mathbf{T} = \{T_1, \dots, T_m\}\) over an alphabet \(\Sigma\), create \(m-1\) new symbols \(\#_1, \dots \#_{m-1} \not\in\Sigma\) and use them to stitch the strings together as \(T = T_1 \#_1 T_2 \#_2 \dots \#_{m-1} T_m\). Then, \(T\) has a \(k\)-attractor of size \(p+m-1\) if and only if \(\mathbf{T}\) has a \(k\)-set attractor of size \(p\), because any attractor for \(T\) has to mark each position of the unique delimiter symbols. The remaining markings induce a \(k\)-set attractor for \(\mathbf{T}\), and a \(k\)-set attractor combined with those markings yields a \(k\)-attractor for \(T\). 
    
    For the other direction, given a string \(T\) just use the singleton \(\{T\}\) as input for \(k\)-set attractor problem.
\end{proof}

\Cref{attractor_reductions} shows the combined results of \cite{DBLP:journals/tcs/MantaciRRRS21} and our proof. 
Given an input string \(T\), \(T^*\) or \(\mathbf{T}\) and solution size \(p\) for the \(k\)-attractor, circular \(k\)-attractor or \(k\)-set attractor problem and a function solving one of the problems, potentially a different one, the corresponding entry describes how to modify the input to decide the problem corresponding to the input with the given function.
Transforming a circular string into a set of strings or vice versa is done using strings as an intermediate step.
A variant of the problem on a set of circular strings is also equivalent by transforming each circular string into a string with the given operations, obtaining a set of strings
which can be further transformed as desired.

\begin{table}[!t]
\renewcommand{\arraystretch}{1.3}
\centering
\begin{tabular}{r|c c c}
input & \(\angles{T, p}\) & \(\angles{T^*, p}\) & \(\angles{\mathbf{T}, p}\)\\
\hline
\hline
string & \(\angles{T, p}\) & \(\angles{T^*T^*T^*, p}\) 
& \(\angles{T_1 \#_1 T_2 \dots \#_{m-1} T_m , p+m-1}\)  \\
\hline
circular & \(\angles{T\#, p+1}\) & \(\angles{T^*, p}\) 
& \(\angles{T_1 \#_1 T_2 \dots \#_{m-1} T_m \#_m, p+m}\) \\
\hline
set & \(\angles{\{T\}, p}\) & \(\angles{\{T^*T^*T^*\}, p}\) 
& \(\angles{\mathbf{T}, p}\) \\
\end{tabular}
\bigskip
\caption{Equivalence of different attractor problems.}
\label{attractor_reductions}
\end{table}

\section{The Colorful Edge Cover Problem}
\label{sec:colorful}

The key idea behind colorfulness is to extend a problem in P, i.e. edge cover, by a color function on the set of solution elements to raise the complexity to NP-complete.
Any solution to the initial problem must then satisfy the additional condition that each color is represented. The colors can be used to model guessing a certificate for an NP-complete problem, and the structure of the initial problem is used to verify that certificate.

We define an edge coloring on a set of colors \(\mathcal{C}\) as a surjective function \(\col: E \to \mathcal{C}\).
It is required that \(\col\) is surjective to avoid trivially unsolvable instances. 
Note that this is different from the edge coloring problem where the colors are subject to the constraint that no two edges of the same color are adjacent.
The set of edges \(E\) is allowed to contain self-loops, i.e., an edge of the form \(\{v,v\}\). 
Normally, self-loops are not of interest for the edge cover problem, because they only cover one vertex, making every other adjacent edge more desirable. 
However, the additional colorfulness constraint can make the self-loops necessary as part of an optimal solution. 

\begin{definition}[Colorful Edge Cover] 
  For an undirected edge-colored graph \(G=(V,E,\col)\), with \(\col: E \to \mathcal{C}\),
  a subset \(E' \subseteq E\) is called a colorful edge cover of \(G\),
  if for each vertex \(v \in V\) there is an edge \(\{v,w\} \in E'\),
  and for each color \(c \in \mathcal{C}\) there is
  an edge \(e \in E'\) with \(\col(e) = c\).
\end{definition}

By minimum colorful edge cover, we denote the optimization problem of finding a smallest colorful edge cover. 
The set \(\{\angles{G, p} : G \text{ has a colorful edge cover of size } p\}\) defines the corresponding decision problem, the colorful edge cover problem. 
Before we show its hardness, we prove that any algorithm that solves the minimum colorful edge cover problem on simple graphs, also solves the problem on graphs with self-loops. 
We achieve this by constructing a gadget of constant size replacing all self-loops. 

\begin{lemma}
    An algorithm solving the colorful edge cover problem on simple graphs also solves the colorful edge cover problem on graphs with self-loops.
\end{lemma}
\begin{proof}
    Self-loops can be simulated by a gadget consisting of two new vertices and a new color.
    Introduce a new color \(b\) and two new vertices \(x,y\) that are connected by an edge of color \(b\).
    For any vertex \(v\) with a self-loop of color \(a\), instead connect \(v\) and \(x\) with color \(a\).
    Any valid solution contains the unique edge of color \(b\) such that \(x\) and \(y\) are always covered. 
    Choosing the edge \(\{v, x\}\) then only covers \(v\) and color \(a\), which is the same behaviour as for the self-loop.
    The resulting graph does not contain self-loops and has a colorful edge cover on \(\mathcal{C}+1\) colors of size \(p+1\) if and only if the original graph has a colorful edge cover on \(\mathcal{C}\) colors of size \(p\).
\end{proof}

In the following, we introduce a special, balanced version of the satisfiability problem that enables us to show the NP-hardness of the colorful edge cover problem. 

\begin{definition}[\((3,B2)\)-SAT]
  A Boolean formula \(\phi = \bigwedge_{i=1}^{m} c_i =
  \bigwedge_{i=1}^{m} \bigvee_{j=1}^{L(c_i)} l_{i,j}\) over the variables 
  \(\mathbf{X_n} = \{x_1, \dots, x_n\} \) is a \((3,B2)\)-SAT instance if each clause consists of
  exactly three literals and each literal occurs exactly two times, thus each variable occurs exactly four times.
  We denote the satisfiability problem for \((3,B2)\)-SAT instances by \((3,B2)\)-SAT.
\end{definition}

By a result from Berman, Karpinski and Scott \cite{APXbound}, \((3,B2)\)-SAT is NP-complete.

\begin{theorem}\label{theorem:cec}
  The colorful edge cover problem is NP-complete.
\end{theorem}
\begin{proof}
  The problem is in NP, as a subset \(E' \subseteq E\) forming a colorful edge cover can be encoded and verified in polynomial space and time.
  We show its hardness by a reduction from \((3,B2)\)-SAT. 

  Given a \((3,B2)\)-SAT formula \(\phi\), for each clause \(c_i = \bigvee_{j=1}^{3} l_j\), construct a clause vertex \(c_i\), and for each \(j \in [1,3]\) construct an intermediate vertex \(c_{i,j}\) connected to \(c_i\), and a literal vertex \(l_{i,j}\) connected to \(c_{i,j}\). 
  All those edges are assigned the color 0.
  Further, for each variable \(x_a\), connect the two vertices \(l_{i_1,j_1}, l_{i_2,j_2}\) corresponding to its positive literal with an edge of color \(a\), and the two vertices \(l_{i_3,j_3}, l_{i_4,j_4}\) corresponding to its negative literals with another edge of color \(a\).
  \Cref{reduction:cec} shows the subgraph for a variable \(x_a\).
  Note that the clause vertices of the form \(c_i\) have two more neighbors \(c_{i,j'}, c_{i,j''}\) not depicted here as they also belong to the gadgets of their other literals. 
  The dashed lines form a solution representing a false assignment to the variable \(x_a\), whereas the solid lines represent setting the variable to true.
  We claim that the constructed graph has a colorful edge cover of size \(n + |\phi| = 5n\)
  if and only if \(\phi\) is satisfiable.

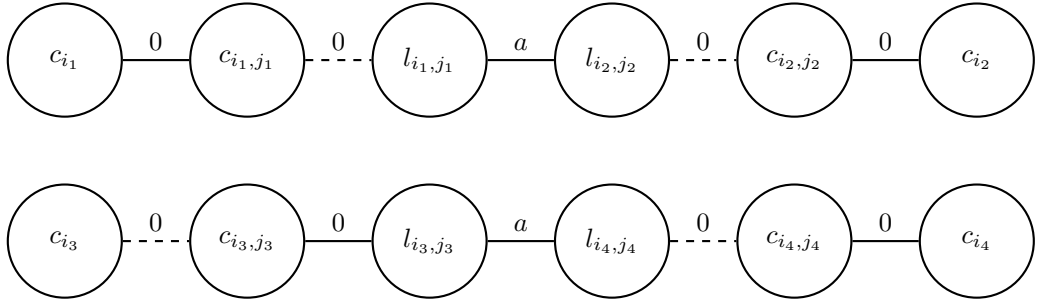
\begin{figure}
    \centering
    \begin{tikzpicture}[node distance={24mm}, thick, main/.style = {draw, circle, minimum size=15mm}]
        \node[main] (1) {$c_{i_1}$}; 
        \node[main] (2)[right of=1] {$c_{{i_1},{j_1}}$}; 
        \node[main] (3)[right of=2] {$l_{{i_1},{j_1}}$}; 
        \node[main] (4)[right of=3] {$l_{{i_2},{j_2}}$}; 
        \node[main] (5)[right of=4] {$c_{{i_2},{j_2}}$}; 
        \node[main] (6)[right of=5] {$c_{i_2}$}; 
        \draw (1) -- (2) node[draw=none,fill=none,midway,above] {\(0\)}; 
        \draw[dashed] (2) -- (3) node[draw=none,fill=none,midway,above] {\(0\)}; 
        \draw(3) -- (4) node[draw=none,fill=none,midway,above] {\(a\)}; 
        \draw[dashed] (4) -- (5) node[draw=none,fill=none,midway,above] {\(0\)}; 
        \draw (5) -- (6) node[draw=none,fill=none,midway,above] {\(0\)}; 
        \node[main] (7)[below of=1] {$c_{i_3}$}; 
        \node[main] (8)[right of=7] {$c_{{i_3},{j_3}}$}; 
        \node[main] (9)[right of=8] {$l_{{i_3},{j_3}}$}; 
        \node[main] (10)[right of=9] {$l_{{i_4},{j_4}}$}; 
        \node[main] (11)[right of=10] {$c_{{i_4},{j_4}}$}; 
        \node[main] (12)[right of=11] {$c_{i_4}$}; 
        \draw[dashed] (7) -- (8) node[draw=none,fill=none,midway,above] {\(0\)}; 
        \draw (8) -- (9) node[draw=none,fill=none,midway,above] {\(0\)}; 
        \draw(9) -- (10) node[draw=none,fill=none,midway,above] {\(a\)}; 
        \draw[dashed] (10) -- (11) node[draw=none,fill=none,midway,above] {\(0\)}; 
        \draw (11) -- (12) node[draw=none,fill=none,midway,above] {\(0\)}; 
    \end{tikzpicture}
    \caption{Gadgets for a variable \(x_a\) in the colorful edge cover reduction.}
    \label{reduction:cec}
\end{figure}

  Assume a colorful edge cover of size \(n + |\phi| = 5n\) exists.
  These costs are always a lower bound for the colorful edge cover, as there has to be
  one chosen edge for each color other than 0, and one edge for each intermediate vertex \(c_{i,j}\),
  as they are pairwise non-adjacent and all their incident edges have color 0.
  
  Another way to see this is to consider the vertices unique to each variable \(x_a\),
  which are four pairs of the form \(c_{i,j}\), \(l_{i,j}\).
  Because only the edges \(\{l_{i_1,j_1}, l_{i_2,j_2}\}\) and \(\{l_{i_3,j_3}, l_{i_4,j_4}\}\) have the color \(a\), it is not possible to cover all the intermediate vertices \(c_{i,j}\) and the color with less than five edges.
  The dashed and solid lines in \Cref{reduction:cec} each refer to one way to cover all unique elements with five edges, while also covering either \(c_{i_1}\) and \(c_{i_2}\) or \(c_{i_3}\) and \(c_{i_4}\). 
  Without loss of generality, we assume that the edges adjacent to the included edge of color \(a\) are not included, so our solution follows this form.
  Then, a clause vertex \(c_i\) is only covered by an edge \(\{c_i, c_{i,j}\}\) if the edge \(\{c_{i,j}, l_{i,j}\}\) is not included.
  In turn, this means that the edge of color \(a\) incident to \(l_{i,j}\) is included, indicating that variable \(X_a\) is assigned a truth value that satisfies \(c_i\).
 
  Given a satisfying assignment of \(\mathbf{X_n}\) for \(\phi\), construct a colorful edge cover for 
  \(G\) as follows. For each variable \(x_a\), consider the two gadgets 
  \(c_{i_1} c_{i_1,j_1} l_{i_1,j_1} l_{i_2,j_2} c_{i_2,j_2} c_{i_2}\) and \(c_{i_3} c_{i_3,j_3} l_{i_3,j_3} l_{i_4,j_4} c_{i_4,j_4} c_{i_4}\).
  If \(x_a\) is positive, choose the edges \[\{c_{i_1}, c_{i_1,j_1}\}, \{l_{i_1,j_1}, l_{i_2,j_2}\}, \{c_{i_2,j_2}, c_{i_2}\}, \{c_{i_3,j_3}, l_{i_3,j_3}\}, \{l_{i_4,j_4} c_{i_4,j_4}\}.\] 
  Note that \(\{l_{i_1,j_1}, l_{i_2,j_2}\}\) has color \(a\).
  If \(x_a\) is negative, choose the edges \[\{c_{i_1,j_1}, l_{i_1,j_1}\}, \{l_{i_2,j_2}, c_{i_2,j_2}\},
  \{c_{i_3}, c_{i_3,j_3}\}, \{l_{i_3,j_3}, l_{i_4,j_4}\}, \{c_{i_4,j_4} c_{i_4}\}.\] 
  Note that \(\{l_{i_3,j_3}, l_{i_4,j_4}\}\) has color \(a\).
  In both cases, we infer a cost of \(5n\). It is clear for all vertices except the clause 
  vertices \(c_i\) that they are covered. Because \(\phi\) is satisfied by the given 
  assignment, any clause \(c_i\) is satisfied by some literal \(l_j\), so by our choice
  the edge \(\{c_i, c_{i,j}\}\) is included, covering \(c_i\). Further, the color 0 is covered
  \(4n\) times, and each other color is covered exactly once.
\end{proof}

\section{NP-Completeness of the 2-Attractor Problem}
\label{sec:2attractor}
We now give a formal definition of a graph interpretation of strings and their substrings of length 2 which we call the \textit{2-substring graph}. 
Each position in a string corresponds to an edge in this graph.
This interpretation was used to show that computing sharp 2-attractors can be done in time \(\mathcal{O}(n\sqrt{n})\) \cite{DBLP:conf/esa/KempaPPR18} by solving the edge cover problem on this graph. 
However, we additionally use symbols to label the edges, yielding instances of the colorful edge cover problem instead. 

\begin{definition}
   Given a set of strings \(\mathbf{T} = \{T_1, \dots T_m\}\) with \(T_i \in \Sigma^{n_i}\), we define the 2-substring graph of \(\mathbf{T}\) by \(G(\mathbf{T}) = (V, E = E_1 \cup E_2 \cup E_3, \sigma)\)
   with \(V = \{xy \in \Sigma^2 \mid xy \text{ is a substring of any }\allowbreak T_i \in \mathbf{T}\} \),
   \begin{align*}
   E_1 =&  \{(xy,yz) \in \left(\Sigma^2\right)^2 \mid xyz \text{ is a substring of any } T_i \in \mathbf{T}\} \\ 
   E_2 =& \{(xy,xy) \in \left(\Sigma^2\right)^2 \mid xy \text{ is the prefix of any } T_i \in \mathbf{T} \} \\ 
   E_3 =& \{(yz,yz) \in \left(\Sigma^2\right)^2 \mid yz \text{ is the suffix of any } T_i \in \mathbf{T} \}
   \end{align*}
   and  
   \(\sigma: E \to \Sigma\) is a labeling function on the edges defined by 
   \[
    \sigma(e) = 
     \begin{cases}
       y, &\text{if } e = (xy,yz) \in E_1 ,\\
       x, &\text{if }  e = (xy,xy) \in E_2 ,\\
       z, &\text{if }  e = (yz,yz) \in E_3 .\\
     \end{cases}\]
\end{definition}

\begin{figure}
    \centering
    \begin{tikzpicture}[node distance={20mm}, thick, main/.style = {draw, circle, minimum size=10mm}]
        \node[main] (1) {$a b$}; 
        \node[main] (2)[right of=1] {$b b$}; 
        \node[main] (3)[right of=2] {$b c$}; 
        \node[main] (4)[right of=3] {$c d$}; 
        
        \node[main] (5)[right of=4] {$c a$}; 
        
        \node[main] (6)[right of=5] {$d e$}; 
        \node[main] (7)[right of=6] {$e c$}; 
        
        \draw[dashed,->](1) to [out=60,in=120,looseness=5] node[draw=none,fill=none,above] {\(a\)} (1) ;
        \draw[->] (1) -> node[draw=none,fill=none,midway,above] {\(b\)} (2); 
        \draw[->](2) to [out=60,in=120,looseness=5] node[draw=none,fill=none,above] {\(b\)} (2) ;
        \draw[->] (2) -- node[draw=none,fill=none,midway,above] {\(b\)} (3); 
        \draw[->] (3) -- node[draw=none,fill=none,midway,above] {\(c\)} (4) ; 
        \draw[dotted, ->](4) to [out=60,in=120,looseness=5] node[draw=none,fill=none,above] {\(d\)} (4) ;
        
        \draw[dashed,->](3) to [out=60,in=120,looseness=5] node[draw=none,fill=none,above] {\(b\)} (3) ;
        \draw[->] (3) to [out=300,in=240] node[draw=none,fill=none,midway,above] {\(c\)} (5);  
        \draw[dotted, ->](5) to [out=60,in=120,looseness=5] node[draw=none,fill=none,above] {\(a\)} (5) ; 
        
        \draw[->] (6) -- node[draw=none,fill=none,midway,above] {\(e\)} (7); 
        \draw[dashed,->](6) to [out=60,in=120,looseness=5] node[draw=none,fill=none,above] {\(d\)} (6) ;
        \draw[dotted, ->](7) to [out=60,in=120,looseness=5] node[draw=none,fill=none,above] {\(c\)} (7) ; 
    \end{tikzpicture}
    \caption{The 2-substring graph for \(\{abbbcd, bca, dec\}\).}
    \label{example:2sub}
\end{figure}
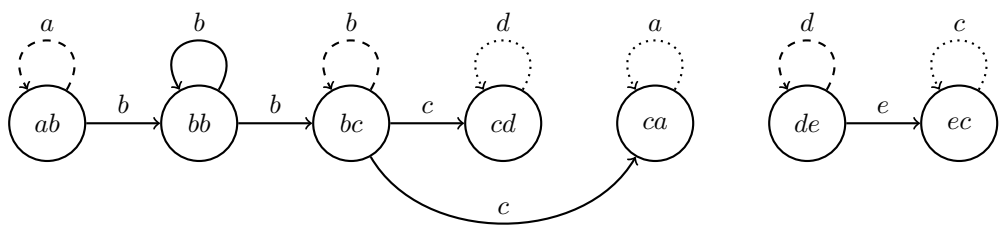

\Cref{example:2sub} shows the 2-substring graph of a set of strings \(\{abbbcd, bca, dec\}\), where the solid lines are from \(E_1\), the dashed lines are the prefix self-loops from \(E_2\) and the dotted lines are the suffix self-loops from \(E_3\). 
Note that the 2-substring graph does not capture words of length 1, as they do not have a substring of length 2 and therefore no associated vertex. 
However, these are not relevant to our problem. 
In a preprocessing step, we can remove the word of length 1 if its symbol also occurs in a longer string, or we have to add it to the solution if this is the only occurrence of the symbol. 

The 2-attractor problem as a special case of the \(k\)-attractor problem is in NP \cite{DBLP:conf/stoc/KempaP18}.
To show NP-hardness, we essentially consider the reduction for colorful edge cover and find an assignment of symbols to the edges such that the resulting graph is a 2-substring graph representing a set of strings. 
The key idea is that this 2-substring graph with the edge labels interpreted as a set of colors has a colorful edge cover of a fixed size \(p\) if and only if the underlying set of strings has a \(k\)-attractor of size \(p\). 
Towards this step, we assign a 2-substring to each vertex created in the proof of \Cref{theorem:cec} and exchange colors with symbols. 
It is vital that each vertex has a unique label so that each vertex uniquely corresponds to a 2-substring. 
For that reason, we need to introduce more symbols and pay attention that they are all covered. 
The 2-substring graph representation of the constructed gadgets is shown in \Cref{reduction:att}. 
Comparing this to \Cref{reduction:cec} shows the similarities in the proofs. 

Note that we do not give a reduction from the colorful edge cover problem to the 2-attractor problem, but use the structure of the colorful edge cover to give a reduction directly from \((3,B2)\)-SAT to the 2-attractor problem on sets.
Combined with \Cref{lemma:setequiv}, this shows the NP-completeness of the 2-attractor problem.

\begin{theorem}\label{theorem:att}
  The 2-attractor problem is NP-complete.
\end{theorem}

\begin{proof}
  We again start with a \((3,B2)\)-SAT formula \(\phi\).
  We create a set of strings based on an alphabet \(\Sigma = \mathbf{C} \cup \mathbf{L} \cup \overline{\mathbf{L}} \cup \mathbf{X} = \{C_1, \dots, C_m\} \cup \{L_1, L_2 , L_3\} \cup \{\overline{L}_1, \overline{L}_2, \overline{L}_3\} \cup \{X_1, \dots, X_n\}\). 
  Note that the symbols for the literals \(\mathbf{L},\overline{\mathbf{L}}\) only indicate the position of the literal in its clause and whether it is negated, but not the clause itself.
  
  We now construct a set of strings \(\mathbf{T}\).
  For every variable \(X_a\) that appears positively in clauses \(c_{i_1}\) and \(c_{i_2}\) at literals \(l_{i_1,j_1}\) and \(l_{i_2,j_2}\) respectively with \(i_1 \leq i_2\) (and \(j_1 < j_2\) if \(i_1 = i_2\)), we add a string \(C_{i_1} C_{i_1} L_{j_1} X_a L_{j_2} C_{i_2} C_{i_2}\).
  For every variable \(X_a\) that appears negatively in clauses \(c_{i_3}\) and \(c_{i_4}\) at literals \(l_{i_3,j_3}\) and \(l_{i_4,j_4}\) respectively with \(i_3 \leq i_4\) (and \(j_3 < j_4\) if \(i_3 = i_4\)), we add a string \(C_{i_3} C_{i_3} L_{j_3} X_a L_{j_4} C_{i_4} C_{i_4}\). 
  We also add six auxiliary strings 
  \(L_1 L_1, \overline{L}_1 \overline{L}_1, L_2 L_2, \overline{L}_2 \overline{L}_2, L_3 L_3, \overline{L}_3 \overline{L}_3\)
  to make sure all symbols \(L_j\) and \(\overline{L}_j\) are covered. 
  
  This construction is shown in \Cref{reduction:att} as subgraphs of the 2-substring graph. Note that the clause vertices of the form \(C_i C_i\) have two more adjacent edges that are not shown here, and the auxiliary strings are not shown.
  The dashed lines form a solution representing a false assignment to the variable \(x_a\), whereas the solid lines represent setting the variable to true. Note that there is always an optimal solution that does not use the self-loops at vertices of the form \(C_i C_i\) which are represented by dotted lines.
  
  The set of 2-substrings that need to be covered consists of \(C_{i} C_{i}\), \(C_{i} L_1\),\(C_{i} L_2\) and \(C_{i} L_3\) for each clause \(C_i\), as well as
  \(L_{j_1} X_a\), \(X_a L_{j_2}\), \(\overline{L}_{j_3} X_a\), \(X_a \overline{L}_{j_4}\) for each variable \(X_a\).
  Note that each of these 2-substrings except for \(C_i C_i\) appears only once in \(\mathbf{T}\), so they need to be covered at that occurrence.
  The auxiliary strings each add one unique 2-substring \(L_j L_j\) or \(\overline{L}_j \overline{L}_j\) for \(j \in \{1,2,3\}\) that also need to be covered in the respective string.
  Of course, all 1-substrings, i.e. \(\Sigma\), also need to be covered.

  We claim that \(\mathbf{T}\) has a 2-set attractor of size \(n+|\phi|+6 = 5n+6\) if and only if \(\phi\) is satisfiable. 
  The proof follows the same arguments as the proof of \Cref{theorem:cec} as the 2-substring graph \(G(\mathbf{T})\) has a colorful edge cover if and only if \(\mathbf{T}\) has a 2-set attractor.

  Assume that a 2-set attractor of size \(5n+6\) for \(\mathbf{T}\) exists. 
  We first show that the \(n\) symbols \(\mathbf{X}\) corresponding to the variables, which all appear twice, are covered exactly once each. 
  The substrings \(L_j L_j\) and \(\overline{L}_j \overline{L}_j\) for \(j \in \{1,2,3\}\) are unique to the six auxiliary strings, so they always induce a cost of six and ensure that all \(L_i\),\(\overline{L}_i\) are covered independent of the remaining strings.
  It is also necessary to expend \(4n\) markings to cover the substrings of the form \(C_i L_j\), \(L_j C_i\), \(C_i \overline{L}_j\) and \(\overline{L}_j C_i\), as each of these are pairwise non-overlapping with each other.
  The remaining \(n\) markings are then needed to cover all \(n\) symbols \(\mathbf{X}\), so each is covered once.

  We can therefore deduce a truth assignment by considering where each symbol \(X_a\) corresponding to a variable is covered. We show that this assignment satisfies \(\phi\).
  Consider any substring \(C_i C_i\). 
  Without loss of generality it is covered at the position adjacent to some \(L_j\) or \(\overline{L}_j\), not at the position corresponding to the start or the end of a string. 
  Therefore, the substrings and markings are of the form \(C_i \underline{C_i} L_j\), \(L_j \underline{C_i} C_i\) or \(C_i \underline{C_i} \overline{L}_j\), \(\overline{L}_j \underline{C_i} C_i\). 
  This also covers the substrings of the form \(C_i L_j\), \(L_j C_i\), \(C_i \overline{L}_j\) and \(\overline{L}_j C_i\). By our counting argument, each of these substrings is covered only once, so \(L_j\) respectively \(\overline{L}_j\) is not marked. 
  To then cover \(L_j X_a\), \(X_a L_j\) or \(\overline{L}_j X_a\), \(X_a\) must be marked, which means the variable \(x_a\) is assigned a truth value such that it satisfies \(c_i\). 
  This holds for all \(c_i\), so \(\phi\) is satisfied.
  
  Assume \(\phi\) is satisfiable by some truth assignment of the variables \(\mathbf{X}\).
  If \(x_a\) is set to true, mark its two strings by
  \[C_{i_1} \underline{C_{i_1}} L_{j_1} \underline{X_a} L_{j_2} \underline{C_{i_2}} C_{i_2}, 
  C_{i_3} C_{i_3}\underline{\overline{L}_{j_3}} X_a \underline{\overline{L}_{j_4}} C_{i_4} C_{i_4}\]
  otherwise mark
  \[C_{i_1} C_{i_1} \underline{L_{j_1}} X_a \underline{L_{j_2}} C_{i_2} C_{i_2}, 
  C_{i_3} \underline{C_{i_3}} \overline{L}_{j_3} \underline{X_a} \overline{L}_{j_4} \underline{C_{i_4}} C_{i_4}.\]
  Both ways of marking the strings cover all eight substrings unique to \(x_a\).
  The solid edges in \Cref{reduction:att} show the included positions in the 2-substring graph for a positively assigned variable, and the dashed lines refer to a negative assignment.
  Each substring \(C_i C_i\) is covered in the substring corresponding to the variable satisfying it.
  Further, each symbol of \(\Sigma\) is also covered. 
  Each \(C_i\) is covered just as \(C_i C_i\) is covered, each \(L_j\) is covered due to the six auxiliary strings, and \(X_a\) is covered by definition of our chosen markings.

  In total, given a formula with \(m\) clauses and \(n\) variables,
  we compute for a set of \(2n+6\) strings whether it has a 2-set attractor of size \(5n+6\).
  Each string has a size of 7 except for the auxiliary strings of size 2, and all strings in total have length \(14n+12\).
  The reduction can be computed naively in time \(\mathcal{O}(n^2)\). 
  By \Cref{lemma:setequiv}, the set of strings can be condensed into a single string of linear size in polynomial time. 
\end{proof}

\begin{figure}
    \centering
    \begin{tikzpicture}[node distance={25mm}, thick, main/.style = {draw, circle, minimum size=10mm}]
        \node[main] (1) {$C_{i_1} C_{i_1}$}; 
        \node[main] (2)[right of=1] {$C_{i_1} L_{j_1}$};
        \node[main] (3)[right of=2] {$L_{j_1} X_a$};
        \node[main] (4)[right of=3] {$X_a L_{j_2}$}; 
        \node[main] (5)[right of=4] {$L_{j_2} C_{i_2}$};
        \node[main] (6)[right of=5] {$C_{i_2} C_{i_2}$};
        
        \draw[->] (1) -> node[draw=none,fill=none,midway,above] {\(C_{i_1}\)} (2); 
        \draw[dashed,->] (2) -- node[draw=none,fill=none,midway,above] {\(L_{j_1}\)} (3);  
        \draw[->] (3) -- node[draw=none,fill=none,midway,above] {\(X_a\)} (4); 
        \draw[dashed,->] (4) -- node[draw=none,fill=none,midway,above] {\(L_{j_2}\)} (5); 
        \draw[->] (5) -- node[draw=none,fill=none,midway,above] {\(C_{i_2}\)} (6); 
        
        \node[main] (7)[below of=1, node distance ={35mm}] {$C_{i_3} C_{i_3}$};
        \node[main] (8)[right of=7] {$C_{i_3} \overline{L}_{j_3}$};
        \node[main] (9)[right of=8] {$\overline{L}_{j_3} X_a$};
        \node[main] (10)[right of=9] {$X_a \overline{L}_{j_4}$};
        \node[main] (11)[right of=10] {$\overline{L}_{j_4} C_{i_4}$};
        \node[main] (12)[right of=11] {$C_{i_4} C_{i_4}$};
        
        \draw[dashed,->] (7) -- node[draw=none,fill=none,midway,above] {\(C_{i_3}\)} (8);
        \draw[->] (8) -- node[draw=none,fill=none,midway,above] {\(\overline{L}_{j_3}\)} (9); 
        \draw[dashed,->] (9) -- node[draw=none,fill=none,midway,above] {\(X_a\)} (10) ; 
        \draw[->] (10) -- node[draw=none,fill=none,midway,above] {\(\overline{L}_{j_4}\)} (11); 
        \draw[dashed,->] (11) -- node[draw=none,fill=none,midway,above] {\(C_{i_4}\)} (12);

        \begin{pgfinterruptboundingbox}
            \draw[dotted,->](1) to [out=60,in=120,looseness=5] node (box) [draw=none,fill=none,above] {\(C_{i_1}\)} (1) ;
            \draw[dotted, ->](6) to [out=60,in=120,looseness=5] node[draw=none,fill=none,above] {\(C_{i_2}\)} (6);
            \draw[dotted,->](7) to [out=60,in=120,looseness=5] node[draw=none,fill=none,above] {\(C_{i_3}\)} (7);
            \draw[dotted,->](12) to [out=60,in=120,looseness=5] node[draw=none,fill=none,above] {\(C_{i_4}\)} (12); 
        \end{pgfinterruptboundingbox}
        \node (dummy) at (box) {};
    \end{tikzpicture}
    \caption{Gadgets for a variable \(x_a\) in the 2-attractor reduction in the 2-substring graph.}
    \label{reduction:att}
\end{figure}
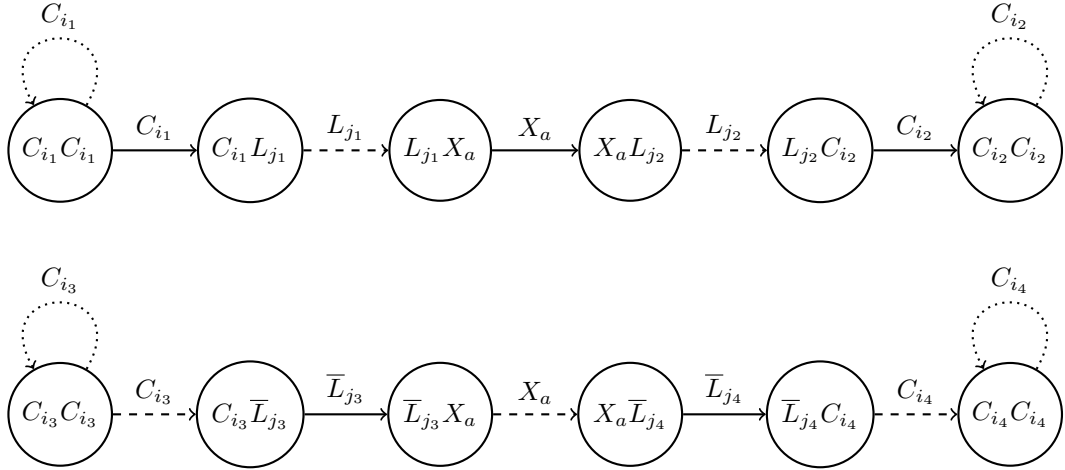

This reduction can be used for all \(k \geq 2\) and for the general attractor problem, as there is a delimiter symbol or unique substring at least every 3 symbols. 
Any valid \(k\)-attractor then needs to put a marking every 3 symbols, such that any possibly uncovered substring has length at most 2.
This enforces that there is a 2-attractor of any fixed size \(p\) if and only if there is a \(k\)-attractor of size \(p\) for any \(k \geq 3\).

\section{APX-Hardness}\label{sec:APX}

Our reduction to prove \Cref{{theorem:att}} also suffices to show the APX-hardness of the \(k\)-attractor problem for \(k \geq 2\), which was shown for \(k \geq 3\) by Kempa and Prezza \cite{DBLP:conf/stoc/KempaP18}.
The same paper also shows containment in APX for constant \(k\).
We also slightly improve the explicit lower bound to which \(k\)-attractor cannot be approximated. 
To this end, we analyse the behavior of the reduction's output string for unsatisfiable formulas.

\begin{lemma}\label{lemma:unsatgap}
    For a MAX-\((3,B2)\)-SAT formula \(\phi\) with \(m\) clauses and an optimal assignment satisfying \(m-u\) clauses, \(u \geq 0\), the optimal attractor for \(T(\phi)\) has size between \(\frac{21}{4}m + 11 + \lceil u/2 \rceil\) and \(\frac{21}{4}m + 11 + u\).
\end{lemma}
\begin{proof}
    A MAX-\((3,B2)\)-SAT formula \(\phi\) with \(m\) clauses contains \(n=(3/4)m\) variables, thus the resulting set of strings is of size \(2m+6\) including auxiliary strings, inducing \(2m+5\) delimiter symbols to combine those strings into a single string \(T\).
    Each variable induces a cost of at least 5 to cover its unique substrings, also each auxiliary string and delimiter induces a cost of 1.
    In total, we obtain a lower bound of \(5n + 6 + 2n + 5 = \frac{21}{4}m + 11\) that can be matched if and only if \(\phi\) is satisfiable with the markings given in the proof of \Cref{theorem:att}.

    If \(\phi\) is not satisfiable, this marking still yields the most efficient way to cover the unique substrings.
    Then, all substrings are covered except for \(C_i C_i\) (and \(C_i\), which will always be covered implicitly) for clauses \(c_i\) that are not covered by the assignment.
    This gives \(u\) many 2-substrings that are not covered, which can be covered one by one using \(u\) additional positions, yielding the upper bound.
    
    Consider a variable \(x_a\) such that in the optimal assignment two of the four clauses it appears in are not satisfied.
    Then, without loss of generality \(x_a\) is set to true and appears negatively in clauses \(c_{i_1}, c_{i_2}\) which are not satisfied.
    The corresponding string is then \(C_{i_1} C_{i_1} \underline{\overline{L}_{j_1}} X_a \underline{\overline{L}_{j_2}} C_{i_2} C_{i_2}\) but remarking it to \(C_{i_1} \underline{C_{i_1}} \overline{L}_{j_1} \underline{X_a} \overline{L}_{j_2} \underline{C_{i_2}} C_{i_2}\) covers both \(C_{i_1} C_{i_1}\) and \(C_{i_2} C_{i_2}\) with just one additional position.
    If this can be done for all clauses, only \(u/2\) additional markings are needed.
    Each variable still has to induce a cost of at least 5 and now induces a cost of at most 6, so this is optimal, otherwise the assignment we started with was not optimal.
\end{proof}

In \cite{APXbound}, Berman, Karpinski and Scott show the APX-hardness of MAX-\((3,B2)\)-SAT by a reduction from the MAX-E3-Lin-2 problem on linear equations, which was studied by Håstad \cite{hastad}.
We extend the given reduction to the \(k\)-attractor problem.

\begin{theorem}\label{theorem:apx}
    For every \(k \geq 2\) and \(0 < \varepsilon < 1\), it is NP-hard to approximate the \(k\)-attractor problem to within an approximation ratio smaller than \((10669-\varepsilon)/10668\).
\end{theorem}
\begin{proof}
    Berman, Karpinski and Scott \cite{APXbound} show that it is NP-hard to distinguish MAX-\((3,B2)\)-SAT instances with \(1016n\) clauses of which at least \((1016 - \epsilon)n\) are satisfiable, and instances with \(1016n\) clauses of which at most \((1015 + \epsilon)n\) are satisfiable.

    Let \(\phi\) be a \((3,B2)\)-SAT formula with \(1016n\) clauses, and \(T(\phi)\) the string resulting from applying the reduction in \Cref{theorem:att} to this formula.
    The formula \(\phi\) has \(1016n \cdot 3/4 = 762n\) variables and thus \(T(\phi)\) consists of \(2 \cdot 762n + 6\) substrings with \(2 \cdot 762n + 5\) delimiter symbols between them.
    If \(\phi\) is satisfiable, \(T(\phi)\) has an optimal \(k\)-attractor of size \((5 \cdot 762n + 6) + (2 \cdot 762n + 5) = 5334n + 11\) by \Cref{lemma:unsatgap}.
    Further, if we can satisfy at least \((1016 - \epsilon)n\) clauses in \(\phi\), we need at most 
    \[(5334 + \epsilon)n + 11 = (10668 + 2\epsilon)n/2 + 11\] 
    many positions to find a 2-attractor for \(T(\phi)\).
    Otherwise, we can satisfy at most \((1015 + \epsilon)n\) clauses and thus need at least 
    \(5334n + \frac{1-\epsilon}{2}n + 11 = (10669 - \epsilon)n/2 + 11 > (10669 - 2\epsilon)n/2 + 11\) 
    many positions to find a 2-attractor for \(T(\phi)\).

    If we could approximate \(k\)-attractor better than \((10669-\varepsilon)/10668\), we could plug in \(\phi\) and see if we get a result within 
    \begin{align*}
        \frac{10669-\varepsilon}{10668}\left(\frac{(10668 + 2\epsilon)}{2}n + 11\right) = 
        \left(\frac{10669}{2} + \frac{10669\epsilon}{10668}  - \frac{\varepsilon}{2} - \frac{2\varepsilon\epsilon}{10668}\right) n + \frac{10669 \cdot 11}{10668}
    \end{align*}
    and accept if and only if this is true.
    For large enough \(n\) and \(\varepsilon\), it holds that
    \begin{align*}
        & & \left(\frac{10669}{2} + \frac{10669\epsilon}{10668}  - \frac{\varepsilon}{2} - \frac{2\varepsilon\epsilon}{10668}\right) n + \frac{10669 \cdot 11}{10668}
        &< \frac{10669-2\epsilon}{2}n + 11 \\
        & \iff & \frac{21337\epsilon}{10668} - \frac{2\varepsilon\epsilon}{10668} + \frac{11}{10668n}
        &< \frac{\varepsilon}{2}\\
    \end{align*}
    and thus we are able to distinguish whether \((1016 - \epsilon)n\) clauses or \((1015 + \epsilon)n\) clauses in \(\phi\) are satisfiable, which is NP-hard.
    Note that \(\varepsilon\) can converge to 0 as \(\epsilon\) converges to 0, thus the proof works for all \(\varepsilon > 0\).
\end{proof}

A similar construction can be used to show that the colorful edge cover problem is APX-hard and cannot be approximated by a factor smaller than \((7621 - \varepsilon)/7620\) for \(0 < \varepsilon < 1\).

\section{Conclusion}
In this paper, we have answered the open problem of the complexity of the \(2\)-attractor problem and discussed the implicit APX-hardness that results from our reduction. 
Additionally, we introduced a more general variation of the \(k\)-attractor problem, the \(k\)-set attractor problem, to make our reduction more convenient. 
Moreover, motivated by the previous reductions between the \(k\)-set cover problem and the \(k\)-attractor problem, we introduced the colorful edge cover problem. 
Although it is not contained in the reduction chain for the main result, it shows where the sharp \(2\)-attractor problem and the \(2\)-attractor problem differ in their complexity. 

In general, adding colorfulness to a problem solvable in polynomial time, like matching, network flow or spanning tree can help to find a different perspective on problems of unresolved complexity that are close to the mentioned problems, but not properly modeled by their unmodified variants. 
In the best case, it helps to understand where the borderline of complexity is and which combination of constraints make a problem hard. 
Colorful problems may also be of interest in the context of parameterized complexity or approximability. 

\bibliography{main}


\end{document}